\usepackage{graphicx}
\usepackage{multicol}
\usepackage{amsmath, amsfonts,amssymb}
\usepackage{enumerate}
\usepackage{mdframed}
\usepackage[ruled, vlined]{algorithm2e}
\usepackage[dvipsnames]{xcolor}
\usepackage{calc}
\usepackage{wrapfig}

\newcommand{\thlist}{\T_{\mathrm{List}}}
\newcommand{\thint}{\T_{\mathrm{Int}}}
\newcommand{\thbvfour}{\T_\mathrm{BV4}}
\newcommand{\thintbvfour}{\T_{\mathrm{IntBV4}}}
\newcommand{\thlistintbvfour}{\T_{\mathrm{ListIntBV4}}}

\newcommand{\siglists}{\Sigma_{\mathrm{List}}}
\newcommand{\sigint}{\Sigma_{\mathrm{Int}}}
\newcommand{\sigbvfour}{\Sigma_{\mathrm{BV4}}}

\newcommand{\sort}[1]{\mathsf{#1}}
\newcommand{\fsy}[1]{\mathsf{#1}}

\newcommand{\listsort}{\sort{list}}
\newcommand{\intsort}{\sort{int}}
\newcommand{\bvsort}[1]{\sort{BV{#1}}}
\newcommand{\elemsort}{\sort{elem}}

\newcommand{\wit}{{\mathit{wit}}}

\newcommand{\smtlib}{SMT-LIB~2\xspace}
\newcommand{\cvcfour}{CVC4\xspace}

\newcommand{\infinite}{{\mathit{inf}}}
\newcommand{\finite}{{\mathit{fin}}}

\newcommand{\qf}{{\it QF}}

\newcommand{\fv}[2]{{\it vars}_{#1}({#2})}
\newcommand{\sorts}[1]{{\cal S}_{#1}}
\newcommand{\func}[1]{{\cal F}_{#1}}

\newcommand{\pred}[1]{{\cal P}_{#1}}

\newcommand{\Ta}{{\cal T}}
\newcommand{\TaEven}{{{\cal T}_{\mathrm{Even}}^{\infty}}}

\newcommand{\Aa}{{\cal A}}
\newcommand{\Ba}{{\cal B}}
\newcommand{\Ca}{{\cal C}}

\newcommand{\set}[1]{\left\{#1\right\}}

\newcommand{\card}[1]{{\left |#1\right |}}

\newcommand{\A}{{\cal A}}

\newcommand{\T}{{\cal T}}

\newcommand{\w}{\wedge}

\newcommand{\distinct}{{\mathit{distinct}}}
\newcommand{\ssi}{S^{{\mathit{si}}}}
\newcommand{\snsi}{S^{{\mathit{nsi}}}}

\newcommand{\ra}{\rightarrow}
\newcommand{\bi}{\begin{itemize}}
\newcommand{\ei}{\end{itemize}}
\newcommand{\be}{\begin{enumerate}}
\newcommand{\ee}{\end{enumerate}}
\newcommand{\bd}{\begin{description}}
\newcommand{\ed}{\end{description}}

\newcommand{\ora}{\overrightarrow}

\newcommand{\suq}{\subseteq}
\newcommand{\st}{{\ |\ }}
\newcommand{\til}{,\dots,}

\usepackage{pgfplotstable,filecontents}
\usepgfplotslibrary{colorbrewer}
\pgfplotsset{compat=1.13}

\usepackage[inline]{enumitem}
\usepackage{hyperref}
\usepackage{cleveref}
\renewenvironment{proof}{\noindent{\it Proof}\/:}{\qed\hspace*{1pt}}

\makeatletter
\def\@citecolor{blue}%
\def\@urlcolor{blue}%
\def\@linkcolor{blue}%

\def\orcidID#1{\smash{\href{http://orcid.org/#1}{\protect\raisebox{-1.25pt}{\protect\includegraphics{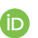}}}}}
\makeatother

\begin{document}
%
\newcommand{\thetitle}{Politeness and Stable Infiniteness: Stronger Together}

\title{\thetitle}
%
%

\hypersetup{
 pdfborder={0 0 0},
 colorlinks=true,
 linkcolor=blue,
 urlcolor=blue,
 citecolor=blue,
 pdfauthor={},
 pdftitle=\thetitle
}

\author{
Ying Sheng\inst{1}\orcidID{0000-0002-1883-2126}
\and
Yoni Zohar\inst{1}\orcidID{0000-0002-2972-6695}
\and
Christophe Ringeissen\inst{2}\orcidID{0000-0002-5937-6059}
\and\\
Andrew Reynolds\inst{3}\orcidID{0000-0002-3529-8682}
\and
Clark Barrett\inst{1}\orcidID{0000-0002-9522-3084}
\and
Cesare Tinelli\inst{3}\orcidID{0000-0002-6726-775X}
}

\institute{Stanford University\and
Universit\'e de Lorraine, CNRS, Inria, LORIA, F-54000 Nancy, France\and
The University of Iowa
}

\maketitle              
\begin{abstract}
We make two contributions to the study of polite combination
in satisfiability modulo theories.
The first is a separation between 
politeness and strong politeness, by
presenting a polite theory that is not strongly polite.
This result shows that proving strong politeness
 (which is often harder than proving politeness) is sometimes needed
in order to use polite combination.
The second contribution is an optimization to the polite combination method,
obtained by borrowing from the Nelson-Oppen method.
The Nelson-Oppen method
is based on guessing arrangements over shared variables.
In contrast,
polite combination requires an arrangement over \emph{all}
variables of the shared sorts.
We show that when using polite combination, if the other theory is stably infinite with respect to 
a shared sort, only
the shared variables of that sort need be considered in arrangements, as in the Nelson-Oppen method.
The time of reasoning about arrangements is exponential in the worst case,
so reducing the number of variables considered
has the potential to improve performance significantly.
We show preliminary evidence for this by demonstrating a speed-up on
a smart contract verification benchmark.
\end{abstract}

\section{Introduction}
\label{sec:intro}
Solvers for satisfiability modulo theories (SMT)~\cite{BT18} are used
in a wide variety of applications.  Many of these applications require determining the satisfiability of formulas with respect to a \emph{combination} of background theories. 
In order to make reasoning about combinations of theories modular and easily extensible,
a combination framework is essential.  Combination frameworks provide mechanisms for
automatically deriving a decision procedure for the combined theories by
using the decision procedures for the individual theories as
black boxes. To integrate a new theory into such a framework, it then suffices to
focus on the decoupled decision procedure for the new theory alone, together with its interface to the generic combination framework.

In 1979, Nelson and Oppen~\cite{NO79} proposed a general framework for combining theories with disjoint signatures.
In this framework, a quantifier-free formula in the combined theory is purified to a conjunction of
formulas, one for each theory. Each pure formula is then sent to a dedicated
theory solver, along with a guessed arrangement
(a set of equalities and disequalities that capture an equivalence relation) of the variables
shared among the pure formulas.
For completeness~\cite{Nelson-CSL-81-10}, this method requires all component theories to be stably infinite.
While many important theories are stably infinite,
some are not, including the widely-used theory
of fixed-length bit-vectors.
To address this issue,
the polite combination method was introduced by Ranise et al.~\cite{RRZ05}, and later
refined by Jovanovic and Barrett~\cite{JBLPAR}. 
In polite combination, one theory must be \emph{polite}, a stronger requirement than stable-infiniteness, but the requirement on the other theory is relaxed: specifically, it need not be stably infinite.
The price for this generality is that unlike the Nelson-Oppen method,
polite combination requires guessing arrangements over \emph{all} variables of certain sorts, not just the shared ones.
At a high level, polite theories have two properties:
smoothness and finite witnessability (see Section~\ref{sec:prelim}).  The polite combination theorem in~\cite{RRZ05} contained an error, which was identified in~\cite{JBLPAR}.  A fix was also proposed in~\cite{JBLPAR}, which relies on stronger requirements for finite witnessability.
Following Casal and Rasga~\cite{Casal2018}, we call this strengthened version 
{\em strong finite witnessability}.  A theory that is both smooth and strongly finitely witnessable is called \emph{strongly polite}.

This paper makes two contributions.
First, we give an affirmative answer to the question of whether politeness and strong politeness are different notions, by
giving
an example of a theory that is polite but not strongly polite.
The given theory is over an empty signature and has two sorts, and 
was originally studied in \cite{Casal2018} in the context of shiny theories.
Here we state and prove the separation of politeness and strong politeness, 
without using shiny theories.
Proving that a theory is strongly polite is harder than proving that it is 
just polite. This result shows that the additional effort is sometimes needed
in order to be able to use the combination theorem from \cite{JBLPAR}.
We show
that for empty signatures, at least two sorts are needed to present a polite
theory that is not strongly polite.
However, for the empty signature with only one sort, there is a finitely witnessable theory that is not strongly finite witnessable.
Such a theory cannot be smooth.

Second, we explore different polite combination scenarios,
where additional information is known about the theories being combined.
In particular, we improve the polite combination method for the case
where one theory is strongly polite w.r.t.~a set $S$ of sorts
and the other is stably infinite w.r.t.~a subset $S'\subseteq S$ of the sorts.
For such cases, we show that it is possible to perform 
Nelson-Oppen combination for $S'$ and polite combination
for $S\setminus S'$. %
This means that for the sorts in $S'$, only shared variables need to be considered for the guessed arrangement,
which can considerably reduce its size.
We also show that the set of shared variables can be reduced for a couple of other variations of
conditions on the theories.
Finally, we present a preliminary case study using
a challenge benchmark from a smart contract verification application.
 We show that the reduction of shared variables is evident
and significantly improves the solving time.  
Verification of smart contracts using SMT, and in particular the analyzed benchmark,
are the main motivation behind the second contribution of this paper.

%
%

\paragraph{Related Work:}
Polite combination is part of a more general effort to replace
the stable infiniteness symmetric condition in the Nelson-Oppen approach 
with a weaker condition.
Other examples of this effort include the notions of
\emph{shiny}~\cite{TZ05}, \emph{parametric}~\cite{DBLP:conf/tacas/KrsticGGT07}, 
and
\emph{gentle}~\cite{DBLP:conf/frocos/Fontaine09} theories.
Gentle, shiny and polite theories can be combined \`a la Nelson-Oppen with any arbitrary theory.
Shiny theories were introduced by Tinelli and Zarba~\cite{TZ05} as a class of mono-sorted theories.
Based on the same principles as shininess, politeness is particularly well-suited to deal with theories expressed in many-sorted logic. 
Polite theories were introduced by Ranise et al.~\cite{RRZ05}
to provide a more effective combination approach
compared to parametric and shiny theories, the former requiring solvers to reason
about cardinalities and the latter
relying on expensive computations of
minimal cardinalities of models.
Shiny theories were extended to many-sorted signatures in~\cite{RRZ05}, where there is a sufficient condition for their equivalence with polite theories.
For the mono-sorted case, a sufficient condition for the equivalence of shiny theories
and strongly polite theories was given by Casal and Rasga~\cite{casal2014equivalence}.
In later work~\cite{Casal2018}, the same authors proposed a generalization of shiny theories 
to many-sorted signatures different from the one in~\cite{RRZ05}, and proved that it
is equivalent
to strongly polite theories with a decidable quantifier-free fragment.
%
The strong politeness of the theory of algebraic datatypes \cite{BST07} was proven in
\cite{DBLP:conf/cade/0007ZRLFB20}. That paper also introduced {\em additive witnesses}, that provided a sufficient condition for a polite theory to be also strongly polite. In this paper we present a theory that is polite but not strongly polite. In accordance with \cite{DBLP:conf/cade/0007ZRLFB20}, the witness that we provide for this theory is not additive.

The paper is organized as follows.
\Cref{sec:prelim} provides the necessary notions from
first-order logic and polite theories.
\Cref{sec:strong_vs_weak} discusses the difference between politeness 
and strong politeness and shows they are not equivalent.
\Cref{sec:polite_nelson_oppen} gives the improvements for the combination process under certain conditions,
and 
\Cref{sec:casestudy} demonstrates the effectiveness of these improvements 
for a challenge benchmark.
\footnote{
\begin{conf}
Due to space constraints,
some proofs are omitted. They can be found in an extended version at \url{https://arxiv.org/abs/2104.11738}.\end{conf}\begin{report}The appendix contains proofs that are omitted from the main text.\end{report}}

\section{Preliminaries}
\label{sec:prelim}
\subsection{Signatures and Structures}
We briefly review the usual definitions of
many-sorted first-order logic with equality (see
\cite{enderton2001mathematical,TinZar-JELIA-04}
for more details).
A {\em signature} $\Sigma$ consists of
a set $\sorts{\Sigma}$ (of {\em sorts}),
a set $\func{\Sigma}$ of function symbols,
and a set $\pred{\Sigma}$ of predicate symbols.
We assume $\sorts{\Sigma}$, $\func{\Sigma}$ and $\pred{\Sigma}$ 
are countable.
Function symbols have arities of the form $\sigma_{1}\times\ldots\times\sigma_{n}\ra\sigma$,
and predicate symbols have arities of the form $\sigma_{1}\times\ldots\times\sigma_{n}$,
with $\sigma_{1}\til\sigma_{n},\sigma\in\sorts{\Sigma}$.
For each sort $\sigma\in\sorts{\Sigma}$,
$\pred{\Sigma}$ includes an {\em equality symbol} $=_{\sigma}$ of arity
$\sigma\times\sigma$.
We denote it by $=$ when $\sigma$ is clear from context.
When $=_{\sigma}$ are the only symbols in $\Sigma$, we say that
$\Sigma$ is {\em empty}.
If two signatures share no symbols except $=_{\sigma}$ we call them
{\em disjoint}.
We assume an underlying countably infinite set of variables for 
each sort.
Terms, formulas, and literals are defined in the usual way.
For a $\Sigma$-formula $\phi$ and a sort $\sigma$, we denote the set of 
free variables in $\phi$ of sort $\sigma$ by $\fv{\sigma}{\phi}$.
This notation naturally extends to $\fv{S}{\phi}$ when $S$ is a set of sorts.
$\fv{}{\phi}$ is the set of all free variables in $\phi$.
We denote by $\qf(\Sigma)$ the set of quantifier-free $\Sigma$-formulas.
%

A {\em $\Sigma$-structure} is a many-sorted structure that provides semantics for the symbols in $\Sigma$ (but not for variables).
It consists of a {\em domain} $\sigma^{\Aa}$ 
for each sort $\sigma\in\sorts{\Sigma}$, 
an interpretation $f^{\Aa}$ for every $f\in\func{\Sigma}$, as well as
an interpretation  $P^{\Aa}$ for every $P\in\pred{\Sigma}$.
We further require that  $=_{\sigma}$ be interpreted
as the identity relation over $\sigma^{\Aa}$ for every $\sigma\in\sorts{\Sigma}$.
A {\em $\Sigma$-interpretation} $\Aa$ is an extension of
a $\Sigma$-structure with interpretations for some set of variables.  
For any $\Sigma$-term $\alpha$, $\alpha^{\Aa}$
denotes the interpretation of $\alpha$ in $\Aa$.
When $\alpha$ is a set of $\Sigma$-terms, $\alpha^{\Aa}=\set{x^{\Aa}\mid x\in \alpha}$.
Satisfaction is defined as usual.  
$\Aa \models \varphi$ denotes that $\Aa$ satisfies $\varphi$.

A $\Sigma$-theory $\Ta$ is a class of all $\Sigma$-structures 
that satisfy some set $Ax$ of $\Sigma$-sentences. 
For each such set $Ax$, we say that $\Ta$ is {\em axiomatized} by $Ax$.
%
%
%
A $\Sigma$-interpretation 
whose variable-free part is in $\Ta$ is called
a $\Ta$-interpretation.
A $\Sigma$-formula $\phi$ is $\Ta$-satisfiable if $\Aa\models\phi$ for some $\Ta$-interpretation $\Aa$.
A set $A$ of $\Sigma$-formulas is $\Ta$-satisfiable if $\Aa\models\phi$ for 
every $\phi\in A$.
Two formulas $\phi$ and $\psi$ are {\em $\Ta$-equivalent} if they are satisfied by the same $\Ta$-interpretations.

Note that for any class $\Ca$ of $\Sigma$-structures there is a
theory $\Ta_{\Ca}$ that {\em corresponds} to it, with the same satisfiable formulas:
the $\Sigma$-theory axiomatized by the set $Ax$ of $\Sigma$-sentences
that are satisfied in every structure of $\Ca$.
In the examples that follow, we define theories $\Ta_{\Ca}$ implicitly 
by specifying only the class $\Ca$,
as done in the \smtlib standard \cite{SMTLib2017}. 
This can be done without loss of generality.\begin{report}\footnote{\sf For further discussion on this point, see \Cref{ap:sticky}.}\end{report}

\begin{example}
\label{ex:lists}
Let $\siglists$ be a signature of finite lists containing the sorts
 $\elemsort_1$, $\elemsort_2$, and $\listsort$,
 as well as the function symbols 
 $\fsy{cons}$ of arity $\elemsort_1\times\elemsort_2\times\listsort\to\listsort$,
 $\fsy{car}_1$ of arity $\listsort\to\elemsort_1$,
 $\fsy{car}_2$ of arity $\listsort\to\elemsort_2$,
 $\fsy{cdr}$ of arity $\listsort\to\listsort$,
 and 
 $\fsy{nil}$ of arity $\listsort$.
 The $\siglists$-theory $\thlist$ corresponds to an
\smtlib theory of algebraic datatypes~\cite{SMTLib2017,BST07}, where
 $\elemsort_1$ and $\elemsort_2$ are interpreted as some sets (of ``elements"),
 and $\listsort$ is interpreted as finite lists
 of pairs of elements, one from $\elemsort_1$ and the other
 from $\elemsort_2$.
 $\fsy{cons}$ is a list constructor that takes
 two elements and a list, and inserts the two elements
 at the head of the list.
 The pair 
 $(\fsy{car}_1(l),\fsy{car}_2(l))$ is the first entry in $l$,
 and $\fsy{cdr}(l)$ is the list obtained from $l$ by
 removing its first entry.
  $\fsy{nil}$ is the empty list.
\qed
\end{example}

\begin{example}
\label{ex:intbv}
	The signature $\sigint$ includes a single sort
	$\intsort$, all numerals $0,1,\ldots$,
	the function symbols $+$, $-$ and $\cdot$
	of arity $\intsort\times\intsort\to\intsort$
	and the predicate symbols
	$<$ and $\leq$
	of arity $\intsort\times\intsort$.
	The $\sigint$-theory $\thint$ corresponds to integer arithmetic in \smtlib,
	and the interpretation of the symbols is the same as in the standard
	structure of the integers.
	The signature $\sigbvfour$ includes a single sort $\bvsort{4}$
and various function
	and predicate symbols for reasoning about bit-vectors
	 of length $4$ (such as $\&$ for bit-wise $and$,
	constants of the form $0110$, etc.).
	The $\sigbvfour$-theory $\thbvfour$ corresponds to \smtlib bit-vectors
	of size 4, with the 
	expected semantics of constants and operators.
	\qed
\end{example}

Let $\Sigma_{1},\Sigma_{2}$ be signatures,
$\Ta_{1}$ a $\Sigma_{1}$-theory, and $\Ta_{2}$ a $\Sigma_{2}$-theory.
The \emph{combination} of $\Ta_{1}$ and $\Ta_{2}$, denoted 
$\Ta_{1}\oplus \T_{2}$,
 consists of all $\Sigma_{1}\cup\Sigma_{2}$-structures $\Aa$,
 such that $\Aa^{\Sigma_{1}}$ is in $\Ta_{1}$ and $\Aa^{\Sigma_{2}}$ is in $\Ta_{2}$,
where $\Aa^{\Sigma_{i}}$ is the reduct of $\Aa$ to $\Sigma_{i}$ for
$i\in\set{1,2}$.

\begin{example}
\label{ex:thlistintbvfour}
Let $\thintbvfour$ be $\thint\oplus\thbvfour$.
It is the combined theory
of integers and bit-vectors. It has all the sorts and 
operators from both theories.
If we rename the sorts $\elemsort_1$ and $\elemsort_2$ of $\siglists$ 
to $\intsort$ and $\bvsort{4}$, respectively, we can obtain
a theory
$\thlistintbvfour$ defined as 
$\thintbvfour\oplus\thlist$.
This is the theory of lists of pairs, where
each pair consists of an integer and a bit-vector of size $4$.
\qed
\end{example}

%

The following definitions and theorems 
will be useful in the sequel.

\begin{theorem}[Theorem 9 of~\cite{TinZar-JELIA-04}]
\label{thm:ordered-sl}
Let $\Sigma$ be a signature,  and
$A$ a set of $\Sigma$-formulas that is satisfiable.
Then there exists an interpretation $\Aa$ that satisfies $A$,
in which $\sigma^{\Aa}$ is countable whenever it is infinite.\footnote{In~\cite{TinZar-JELIA-04} this was proven more generally, for ordered sorted logics.}
\end{theorem}

The following theorem from~\cite{JBLPAR} is a variant of a theorem from~\cite{DBLP:conf/jelia/TinelliZ04}.

\begin{definition}[Arrangement]
Let 
$V$ be a finite set of variables whose sorts are in $S$ and let 
$\set{V_{\sigma}\st \sigma\in S}$ be a partition of $V$ such that $V_{\sigma}$ is the set of variables of sort $\sigma$ in $V$.
A formula $\delta$ is an {\em arrangement of $V$}
if 
\[
  \delta = \bigwedge_{\sigma\in S}(\bigwedge_{(x,y)\in E_{\sigma}}(x=y) \ \w 
           \bigwedge_{x,y\in V_\sigma,(x,y)\notin E_{\sigma}}(x\neq y))\ ,
\]
 where $E_{\sigma}$ is some equivalence relation over $V_{\sigma}$ for each $\sigma\in S$.

\end{definition}

\begin{theorem}[Theorem 2.5 of~\cite{JBLPAR}]
\label{oldbutgood}
For $i=1,2$, let
$\Sigma_i$ be disjoint signatures,
$S_i=\sorts{\Sigma_i}$ with $S=S_1\cap S_2$,
$\Ta_i$ be a $\Sigma_i$-theory,
$\Gamma_i$ be a set of $\Sigma_i$-literals, 
and
$V=\fv{}{\Gamma_1}\cap \fv{}{\Gamma_2}$.
If there exist a $\Ta_1$-interpretation $\Aa$,
a $\Ta_2$ interpretation $\Ba$,
and an arrangement $\delta_V$ of $V$ such that:
\begin{enumerate*}
    \item $\Aa\models\Gamma_1\cup\delta_V$;
    \item $\Ba\models\Gamma_2\cup\delta_V$; and
    \item $|A_\sigma|=|B_\sigma|$ for every $\sigma\in S$,
\end{enumerate*}
then $\Gamma_{1}\cup \Gamma_{2}$ is $\Ta_1\oplus \Ta_2$-satisfiable.
\end{theorem}

\subsection{Polite Theories}
\label{prelim:polite}
We now give the background definitions necessary for both Nelson-Oppen and polite combination.  In what follows, $\Sigma$ is an arbitrary (many-sorted) signature, 
$S\suq\sorts{\Sigma}$,
and $\Ta$ is a $\Sigma$-theory.  We start with stable infiniteness and smoothness.

\begin{definition}[Stably Infinite]
$\Ta$ is {\em stably infinite with respect to $S$}
if every quantifier-free $\Sigma$-formula 
that is $\Ta$-satisfiable is also satisfiable
in a $\Ta$-interpretation $\Aa$ in which $\sigma^{\Aa}$
is infinite for every $\sigma\in S$.
\end{definition}

\begin{definition}[Smooth]
$\Ta$ is \,{\em smooth} w.r.t.\ $S$ if for every quantifier-free formula $\phi$,  $\Ta$-interpretation $\Aa$ that satisfies $\phi$, and function 
$\kappa$ from $S$ to the class of cardinals such that $\kappa(\sigma)\geq\card{\sigma^{\Aa}}$ for every $\sigma\in S$, 
there exists a $\Ta$-interpretation $\Aa'$ that satisfies $\phi$ with 
$\card{\sigma^{\Aa'}}=\kappa(\sigma)$ for every $\sigma\in S$.
\end{definition}

%

\noindent
We identify
singleton sets with their single elements when there is no ambiguity 
(e.g., when saying that a theory is smooth w.r.t.\ a sort $\sigma$).

We next define politeness and related concepts, following the presentation in~\cite{DBLP:conf/cade/0007ZRLFB20}.
Let $\phi$ be a quantifier-free $\Sigma$-formula.
A $\Sigma$-interpretation $\Aa$ {\em finitely witnesses $\phi$ for $\Ta$ w.r.t.\ $S$} (or, is a {\em finite witness of $\phi$ for $\Ta$ w.r.t.\ $S$}), if 
$\Aa\models\phi$ and 
$\sigma^{\Aa}=\fv{\sigma}{\phi}^{\Aa}$
 for every $\sigma\in S$.
We say that $\phi$ is 
{\em  finitely witnessed for $\Ta$ w.r.t.\ $S$} if
it is either $\Ta$-unsatisfiable or has a finite witness for $\Ta$ w.r.t.\ $S$.
We say that $\phi$ is 
{\em strongly finitely witnessed for $\Ta$ w.r.t.\ $S$} if
$\phi\w\delta_{V}$ is  finitely witnessed for $\Ta$ w.r.t.\ $S$ for every arrangement $\delta_{V}$ of $V$, where $V$ is any set of variables whose sorts are in $S$.
A function $\wit : \qf(\Sigma) \ra \qf(\Sigma)$ is
a {\em (strong) witness for $\Ta$ w.r.t.\ $S$} if for every $\phi\in \qf(\Sigma)$ we have that:
\begin{enumerate*}
\item $\phi$ and $\exists\,\ora{w}.\:\wit(\phi)$ are $\Ta$-equivalent for $\ora{w}=\fv{}{\wit(\phi)}\setminus\fv{}{\phi}$; and
\item $\wit(\phi)$ is (strongly) finitely witnessed for $\Ta$ w.r.t.\ $S$.
\end{enumerate*}
$\Ta$ is {\em (strongly) finitely witnessable} w.r.t.\ $S$ if there exists a 
computable (strong) witness for $\Ta$ w.r.t.\ $S$.
$\Ta$ is {\em (strongly) polite w.r.t.\ $S$} if it is smooth and (strong\-ly) finitely witnessable w.r.t.\ $S$.


\section{Politeness and Strong Politeness}
\label{sec:strong_vs_weak}
In this section we study the difference between politeness and strong politeness.
Since the introduction of strong politeness in~\cite{JBLPAR}, it has been unclear
whether it is strictly stronger than politeness,
that is, whether there exists a theory that is polite but not strongly polite.
We present an example of such a theory, answering the open question affirmatively. This result is followed by further analysis of notions related to politeness.
This section is organized as follows.
In \Cref{sec:witvsstrongwit} we reformulate an example given in
\cite{JBLPAR}, showing that there are witnesses that are not
strong witnesses. 
We then present a polite theory that is not strongly polite in
\Cref{sec:politenostrong}.
The theory is over a signature with two sorts but is otherwise empty.
We show in \Cref{sec:singlesort} that politeness and strong politeness are equivalent for empty signatures with a single sort. 
Finally, we show in \Cref{sec:singlesortfin} that this equivalence does not hold for finite witnessability alone.
%

%
%

\subsection{Witnesses vs. Strong witnesses}
\label{sec:witvsstrongwit}
In \cite{JBLPAR} an example was given
for a witness that is not strong.
We reformulate this example in terms of the notions
that are defined in the current paper, that is,
witnessed formulas are not the same as
strongly witnessed formulas (\Cref{weaknonweak1}),
and witnesses are not the same as strong witnesses 
(\Cref{nonweakwitfunc}).

\begin{example}
\label{weaknonweak1}
Let $\Sigma_{0}$ be a signature with a single sort $\sigma$ and no function or predicate symbols, and let $\Ta_{0}$ be a $\Sigma_{0}$-theory 
consisting of all $\Sigma_{0}$-structures with at least two elements.
Let $\phi$ be the formula $x=x\land w=w$.
This formula is finitely witnessed for $\Ta_{0}$ w.r.t.\ $\sigma$, but not strongly.
Indeed, for
$\delta_{V}\equiv(x=w)$,
$\phi\w\delta_{V}$ is not finitely witnessed for $\Ta_{0}$ w.r.t.\ $\sigma$: 
a finite witness would be required to have only a single element
and would therefore not be a $\Ta_{0}$-interpretation.
\qed
\end{example}

\noindent
The next example shows that witnesses and strong witnesses are not equivalent.

\begin{example}
\label{nonweakwitfunc}
Take $\Sigma_{0}$, $\sigma$, and $\Ta_{0}$ as in \Cref{weaknonweak1}, and define 
$\wit(\phi)$ as the function $(\phi\ \w\ w_{1}=w_{1}\ \w\ w_{2}=w_{2})$ for fresh $w_{1},w_{2}$.
The function is a  witness for $\Ta_{0}$ w.r.t.\ $\sigma$.
However, it is not a strong witness for $\Ta$ w.r.t.~$\sigma$. 
 \qed
\end{example}

Although the theory $\Ta_{0}$ in the above examples does serve to distinguish formulas and witnesses that are and are not strong, it cannot be used to do the same for theories themselves.  
This is because $\Ta_{0}$ is, in fact, strongly polite, via a different witness function.
\begin{example}
\label{ex:strong}
 The function $\wit'(\phi)=(\phi\w w_{1}\neq w_{2})$, for some $w_{1},w_{2}\notin\fv{\sigma}{\phi}$, \emph{is} a strong witness for $\Ta_{0}$ w.r.t.\ $S$,
 as proved in~\cite{JBLPAR}.
\qed
\end{example}

\noindent
A natural question, then, is whether there is a theory that can separate the two notions of politeness.
The following subsection provides an affirmative answer.

\subsection{A Polite Theory that is not Strongly Polite}
\label{sec:politenostrong}

Let $\Sigma_{2}$ be a signature with two sorts
$\sigma_{1}$ and $\sigma_2$ and no function or
predicate symbols (except $=$).
Let $\Ta_{2,3}$ be the $\Sigma_{2}$-theory from \cite{Casal2018},
consisting of all $\Sigma_{2}$-structures $\A$
such that
either $\card{\sigma_{1}^{\Aa}}= 2 \wedge \card{\sigma_{2}^{\Aa}}\ge\aleph_0$ 
or $\card{\sigma_{1}^{\Aa}}\geq 3 \wedge \card{\sigma_{2}^{\Aa}}\geq 3$~\cite{Casal2018}.%
\footnote{In~\cite{Casal2018}, the first condition is written $\card{\sigma_{1}^{\Aa}}\ge 2$.  We use equality as this is equivalent and we believe it makes things clearer.}

$\Ta_{2,3}$ is polite, but is not strongly polite.
Its smoothness is shown by extending any given structure with
new elements as much as necessary. 
\begin{lemma}
\label{lem:ta23smooth}
$\Ta_{2,3}$ is smooth w.r.t.~$\set{\sigma_{1},\sigma_{2}}$.
\end{lemma}

For finite witnessability, consider the function $\wit$ defined as follows:
\begin{equation}
\wit(\phi):=\phi \wedge x_1=x_1\w x_2=x_2\w x_3=x_3\w  y_1=y_1\w y_2=y_2\w y_3=y_3
\end{equation}
for
fresh variables $x_1$, $x_2$, and $x_3$ of sort $\sigma_{1}$ and
$y_1$, $y_2$, and $y_3$ of sort $\sigma_{2}$.
It can be shown that $\wit$ is a witness for $\Ta_{2,3}$
but there is no strong witness for it.

\begin{lemma}
\label{lem:ta23fw}
$\Ta_{2,3}$ is finitely witnessable w.r.t.~$\set{\sigma_{1},\sigma_{2}}$.
\end{lemma}

\noindent

\begin{lemma}
\label{lem:23nonfw}
$\Ta_{2,3}$ is not strongly finitely witnessable w.r.t.~$\set{\sigma_{1},\sigma_{2}}$.
\end{lemma}

\Cref{lem:ta23smooth,lem:ta23fw,lem:23nonfw} have shown
that $\Ta_{2,3}$ is polite but is not strongly polite.
And indeed, using the polite combination method from \cite{JBLPAR} with this theory
can cause problems.
Consider the theory $\Ta_{1,1}$ that consists of all $\Sigma_2$-structures $\Aa$
such that $\card{\sigma_1^{\Aa}}=\card{\sigma_2^{\Aa}}=1$.
Clearly, $\Ta_{1,1}\oplus \Ta_{2,3}$ is empty, and hence no formula is
$\Ta_{1,1}\oplus\Ta_{2,3}$-satisfiable.
However, denote the formula $true$ by $\Gamma_1$ 
and the formula $x=x$ by $\Gamma_2$ for some variable $x$ of sort $\sigma_{1}$.
Then $\wit(\Gamma_2)$ is 
$x=x
\wedge
\bigwedge_{i=1}^{3}
x_i=x_i\wedge y_i=y_i$.
Let $\delta$ be the arrangement $x=x_1=x_2=x_3\wedge y_1=y_2=y_3$.
It can be shown that 
$\wit(\Gamma_2)\wedge\delta$ is $\Ta_{2,3}$-satisfiable
and $\Gamma_1\wedge \delta$ is $\Ta_{1,1}$-satisfiable.
Hence the combination method of \cite{JBLPAR} would consider
$\Gamma_1\wedge\Gamma_2$ to be $\Ta_{1,1}\oplus\Ta_{2,3}$-satisfiable, which is impossible.
Hence the fact that $\Ta_{2,3}$ is not strongly polite propagates all the way to the polite combination method.\footnote{
Notice that $\Ta_{2,3}$ can be axiomatized using the following set of axioms,
given the definitions in \Cref{fig:cardfor}:
$
\set{\psi_{\geq 2}^{\sigma_{1}},\psi_{\geq 3}^{\sigma_{2}}}\cup\set{\psi_{=2}^{\sigma_{1}}\ra\neg\psi_{=n}^{\sigma_{2}}\mid n\geq 3}
$}

\begin{figure}[t]
\begin{mdframed}
\vspace{-.8em}
\begin{gather*}
\distinct(x_1\til x_n):=\bigwedge_{1\leq i < j<=n}x_i\neq x_j \\
\psi_{\geq n}^{\sigma}:=\exists x_1\til x_n.\distinct(x_1\til x_n) \\
\psi_{\leq n}^{\sigma}:=\exists x_1\til x_n.\forall y.\bigvee_{i=1}^{n}y=x_i\\
\psi_{=n}^{\sigma}:=\psi_{\geq n}^{\sigma}\wedge\psi_{\leq n}^{\sigma} 
\end{gather*}
\vspace{-2em}
\caption{Cardinality formulas for sort $\sigma$. All variables are assumed to have sort $\sigma$.}	
\label{fig:cardfor}
\end{mdframed}
\end{figure}

\begin{remark}
\label{rem:separation}
An alternative way to separate politeness from strong politeness
using $\T_{2,3}$ can be obtained through shiny theories, as follows.
Shiny theories were introduced in \cite{TZ05} for the
mono-sorted case, and were generalized to
many-sorted signatures in two different ways in 
\cite{Casal2018} and \cite{RRZ05}.
In \cite{Casal2018}, $T_{2,3}$ was introduced as a
theory that is
shiny according \cite{RRZ05}, but not
according to \cite{Casal2018}.
Theorem 1 of \cite{Casal2018} states that their
notion of shininess is equivalent to strong politeness
for theories in which the satisfiability problem for
quantifier-free formulas is decidable.
Since this is the case for $T_{2,3}$, and
since it is not shiny according to \cite{Casal2018},
we get that
$T_{2,3}$ is not strongly polite.
Further, Proposition 18 of \cite{RRZ05} states that
every shiny theory (according to their definition) is polite.
Hence we get that $T_{2,3}$ is polite but not strongly polite.

We have (and prefer) a direct proof based only on politeness, 
without a detour through shininess. Note also 
that \cite{Casal2018} dealt only with strongly polite theories and did not study the weaker notion of polite theories.
In particular, the fact that strong politeness is different from politeness was not stated nor proved there.
\end{remark}

\subsection{The Case of Mono-sorted Polite Theories}
\label{sec:singlesort}

Theory $\Ta_{2,3}$ includes two sorts, but is otherwise empty.
In this section we show that requiring two sorts is essential for separating politeness from strong politeness in otherwise empty signatures. That is, we prove that politeness implies strong politeness
otherwise.
Let $\Sigma_0$ be the signature with a single sort 
$\sigma$ and no function or predicate symbols (except $=$),
We show that smooth $\Sigma_0$-theories have a certain form,
and conclude strong politeness from politeness.

\begin{lemma}
\label{lem:smoothatleast}
Let $\Ta$ be a $\Sigma_0$-theory.
If $\Ta$ is smooth w.r.t.~$\sigma$ and includes a finite structure, 
$\Ta$ is axiomatized by $\psi_{\geq n}^{\sigma}$ from \Cref{fig:cardfor}
for some $n>0$.
\end{lemma}

\begin{proposition}
\label{prop:singlesort}
If $\Ta$ is a $\Sigma_0$-theory that is polite w.r.t.~$\sigma$, then 
it is strongly polite w.r.t.~$\sigma$.
\end{proposition}

\begin{remark}
We again note (as we did in \Cref{rem:separation}) that an alternative way to obtain this result
is via shiny theories,
using~\cite{RRZ05}, which introduced polite theories,
as well as~\cite{10.1007/978-3-642-45221-5_15}, which compared strongly polite theories
to shiny theories in the mono-sorted case.
Specifically, in the presence of a single sort, Proposition 19 of \cite{RRZ05}
states that: 

\noindent
\begin{tabular}{p{2em}p{30em}}
$(\ast)$ & every polite theory
over a finite signature such that it is decidable 
whether a finite structure is a member of the theory, is shiny.
\end{tabular}

\noindent
In turn, Proposition 1 of~\cite{10.1007/978-3-642-45221-5_15}
states that:

\noindent
\begin{tabular}{p{2em}p{30em}}
$(\ast\ast)$ & every shiny theory over a mono-sorted signature with a decidable 
satisfiability problem for quantifier-free formulas, is also strongly polite.
\end{tabular}

\noindent
It can be shown that for every polite $\Sigma_0$-theory
it is decidable whether a finite structure is in the theory. It can also be shown that satisfiability of quantifier-free formulas is decidable for such theories. Using $(\ast)$ and $(\ast\ast)$, we get that in
$\Sigma_0$-theories, politeness implies strong politeness.
Similarly to \Cref{rem:separation}, we prefer a direct route for showing this result, without going through shiny theories.
\end{remark}

\subsection{Mono-sorted Finite witnessability}
\label{sec:singlesortfin}
We have seen that for $\Sigma_0$-theories,
politeness and strong politeness are the same.
Now we show that smoothness is crucial for this equivalence, i.e.,
that there is no such equivalence between finite witnessability
and strong finite witnessability.
%
%
Let $\TaEven$ be the $\Sigma_0$-theory of all $\Sigma_0$-structures
$\Aa$ such that $\card{\sigma^{\Aa}}$ is even or
infinite.\footnote{Notice that $\TaEven$ can be axiomatized using the set 
$\set{\neg\psi_{=2n+1}^{\sigma}\mid n\in\mathbb{N}}$.}
Clearly, this theory is not smooth.
\begin{lemma}
\label{lem:notsmooth}
$\TaEven$ is not smooth w.r.t.~$\sigma$.
\end{lemma}

\newcommand{\even}{\mathit{even}}
We can construct a witness $\wit$ for $\TaEven$ as follows.
Let $\phi$ be a quantifier-free $\Sigma_0$-formula, and 
let
$E$ be the set of all equivalence relations over 
$\fv{}{\phi}\cup\{w\}$ for some fresh variable $w$.  
Let $\even(E)$ be the set of all equivalence relations in $E$ with an even number of equivalence classes.
Then, $\wit(\phi)$ is
$\phi\wedge\bigvee_{e\in \even(E)}\delta_e$, where for each
$e\in \even(E)$, $\delta_e$ is the arrangement induced by $e$:
\begin{equation*}
\bigwedge_{(x,y)\in e}x=y\ \wedge\bigwedge_{x,y\in
\fv{}{\phi}\cup\set{w}\wedge(x,y)\not\in e}x\neq y
\end{equation*}
It can be shown that $\wit$ is indeed a witness, and that $\TaEven$ has no strong witness, similarly to \Cref{lem:23nonfw}.
\begin{lemma}
\label{thm:weakeven}
$\TaEven$ is finitely witnessable w.r.t.~$\sigma$.
\end{lemma}

\begin{lemma}
\label{thm:weaknotstrong}
$\TaEven$ is not strongly finitely witnessable w.r.t.~$\sigma$.
\end{lemma}

\section{A Blend of Polite and Stably-Infinite Theories}
\label{sec:polite_nelson_oppen}

In this section,
we show that the polite combination method can be optimized
to reduce the search space of possible arrangements.
In what follows,
$\Sigma_1$ and $\Sigma_2$ are disjoint signatures,
$S=\sorts{\Sigma_1}\cap\sorts{\Sigma_{2}}$,
$\Ta_1$ is a $\Sigma_1$-theory,
$\Ta_2$ is a $\Sigma_2$-theory,
$\Gamma_1$ is a set of $\Sigma_1$-literals, and
$\Gamma_2$ is a set of $\Sigma_2$-literals.

The Nelson-Oppen procedure reduces 
the $\Ta_1\oplus\Ta_2$-satisfiability of
$\Gamma_1\cup\Gamma_2$ 
to the 
existence of an arrangement $\delta$ over 
the set 
$V=\fv{S}{\Gamma_{1}}\cap\fv{S}{\Gamma_{2}}$, such that
$\Gamma_{1}\cup\delta$ is $\Ta_1$-satisfiable and
$\Gamma_{2}\cup\delta$ is $\Ta_2$-satisfiable.
The correctness of this reduction relies on the fact that both theories are stably infinite w.r.t.~$S$.
In contrast, the polite combination method only requires a condition (namely strong politeness) from one of the theories, 
while the other theory is unrestricted and, in particular, not necessarily stably infinite.
In polite combination, the
$\Ta_1\oplus\Ta_2$-satisfiability of $\Gamma_1\cup\Gamma_2$
is again reduced to the existence of an arrangement $\delta$, but over a different set
$V'=\fv{S}{\wit(\Gamma_{2})}$, such that
$\Gamma_{1}\cup\delta$ is $\Ta_1$-satisfiable and
$\wit(\Gamma_{2})\cup\delta$ is $\Ta_2$-satisfiable, where
$\wit$ is a strong witness for $\Ta_2$ w.r.t.~$S$.
Thus, the flexibility offered by polite combination comes with a price. 
The set $V'$ is potentially larger than $V$ as it contains \emph{all} variables with sorts in $S$ that occur in $\wit(\Gamma_2)$, not just those that also occur in $\Gamma_1$.  
Since the search space of arrangements over a set grows exponentially with its size, this difference can become crucial.
If $\Ta_1$ happens to be stably infinite w.r.t.~$S$, however,
we can fall back to Nelson-Oppen combination and only consider variables that are shared 
by the two sets.
But what if $\Ta_1$ is stably infinite only w.r.t.~to some proper subset $S'\subset S$? Can this knowledge about $\Ta_1$ 
help in finding some set $V''$ of variables between $V$ and $V'$, such that we need only consider arrangements of $V''$?
In this section we prove that this is possible by taking
$V''$ to include only the variables of sorts in $S'$ that are shared between $\Gamma_1$ and $\wit(\Gamma_2)$, and
all the variables of sorts in $S\setminus S'$ that occur
in $\wit(\Gamma_2)$.
We also identify several weaker conditions on $\Ta_2$ that are sufficient for the combination theorem to hold.

\subsection{Refined Combination Theorem}

To put the discussion above in formal terms, we recall the following theorem.

\begin{theorem}[\cite{JBLPAR}]
\label{thm:original}
If $\Ta_2$ is
strongly polite w.r.t.~$S$ with a witness
 $\wit$,
then the following are equivalent:
\begin{itemize*}
    \item[1.] $\Gamma_1\cup \Gamma_2$ is $(T_1\oplus T_2)$-satisfiable;
    \item[2.] there exists an arrangement $\delta_V$ over $V$, such that $\Gamma_1\cup \delta_V$ is $\Ta_1$-satisfiable and $\wit(\Gamma_2)\cup \delta_V$ is $\Ta_2$-satisfiable,
\end{itemize*}
where $V=\bigcup_{\sigma\in S}V_{\sigma}$, and
$V_{\sigma}=\fv{\sigma}{\wit(\Gamma_2)}$ for each $\sigma\in S$.
\end{theorem}

Our goal is to identify general cases in which information regarding $\Ta_1$
can help reduce the size of the set $V$.
We extend the definitions of stably infinite, smooth, and strongly finitely witnessable
to two sets of sorts rather than one. 
Roughly speaking, in this extension, the usual definition is taken for the first set, and some cardinality-preserving constraints are enforced on the second set.

\begin{definition}
\label{newDefs}
Let $\Sigma$ be a signature,
$S_1,S_2$ two disjoint subsets of 
$\sorts{\Sigma}$, and
$\Ta$ a $\Sigma$-theory.

$\Ta$ is {\em (strongly) stably infinite w.r.t.~$(S_1,S_2)$} if for every 
quantifier-free $\Sigma$-formula $\phi$
and $\Ta$-interpretation $\Aa$ satisfying $\phi$,
there exists a $\Ta$-interpretation $\Ba$
such that
$\Ba\models\phi$,
$|\sigma^{\Ba}|$ is infinite for every $\sigma\in S_1$, and
$|\sigma^{\Ba}|\leq |\sigma^{\Aa}|$
($|\sigma^{\Ba}| = |\sigma^{\Aa}|$) for every $\sigma\in S_2$.

$\Ta$ is {\em smooth w.r.t.~$(S_1,S_2)$} if
for every quantifier-free $\Sigma$-formula $\phi$,
$\Ta$-interpretation $\Aa$ satisfying $\phi$, and function
$\kappa$ from $S_1$ to the class of cardinals such that $\kappa(\sigma)\geq\card{\sigma^{A}}$ for each $\sigma\in S_1$,
there exists a $\Ta$-interpretation $\Ba$ that satisfies $\phi$, with
$\card{\sigma^{B}}=\kappa(\sigma)$ for each $\sigma\in S_1$,
and with $\card{\sigma^{B}}$ infinite whenever $\card{\sigma^{\Aa}}$ is infinite
for each $\sigma\in S_2$.

$\Ta$ is {\em strongly finitely witnessable w.r.t.~$(S_1, S_2)$} if
there exists a computable function $\wit: QF(\Sigma)\rightarrow QF(\Sigma)$
such that
for every quantifier-free $\Sigma$-formula $\phi$:
\begin{enumerate*}
\item $\phi$ and $\exists\,\ora{w}.\:\wit(\phi)$ are $\Ta$-equivalent for $\ora{w}=\fv{}{\wit(\phi)}\setminus\fv{}{\phi}$; and
\item for every $\Ta$-interpretation $\Aa$ and arrangement
$\delta$ of any set of variables whose sorts are in $S_1$, 
if $\Aa$ satisfies $\wit(\phi)\wedge\delta$, then
there exists a $\Ta$-interpretation $\Ba$ that
finitely witnesses $\wit(\phi)\wedge\delta$ w.r.t.~$S_1$ and 
for which $\card{\sigma^{\Ba}}$ is infinite whenever $\card{\sigma^{\Aa}}$ is
infinite, for each $\sigma\in S_2$.
\end{enumerate*}
\end{definition}

\noindent
Our main result is the following.

\begin{theorem}
\label{yoniafterclark}
Let $\ssi\subseteq S$ and $\snsi=S\setminus \ssi$.
Suppose $\Ta_1$ is stably infinite w.r.t.~$\ssi$
and one of the following holds:
\begin{enumerate}
\item\label{item:strongsi} $\Ta_2$ is strongly stably infinite w.r.t.~$(\ssi, \snsi)$
and strongly polite w.r.t.~$\snsi$ with a witness $\wit$.
\item\label{item:mediumsi} $\Ta_2$ is stably infinite w.r.t.~$(\ssi, \snsi)$,
smooth w.r.t.~$(\snsi,\ssi)$, and strongly finitely witnessable w.r.t.~$\snsi$ with a witness $\wit$.
\item\label{item:weaksi} $\Ta_2$ is stably infinite w.r.t.~$\ssi$ and
both smooth and strongly finitely-witnessable w.r.t.~$(\snsi, \ssi)$ with a witness $\wit$.
\end{enumerate}
Then the following are equivalent:
\begin{itemize*}
    \item[1.] $\Gamma_1\cup \Gamma_2$ is $(\T_1\oplus \T_2)$-satisfiable;
    \item[2.] There exists an arrangement $\delta_V$ over $V$ such that $\Gamma_1\cup \delta_V$ is $\Ta_1$-satisfiable, and $\wit(\Gamma_2)\cup \delta_V$ is $\Ta_2$-satisfiable,
\end{itemize*}
where $V=\bigcup_{\sigma\in S}V_{\sigma}$, with
$V_{\sigma}=\fv{\sigma}{\wit(\Gamma_2)}$ for every $\sigma\in \snsi$
and $V_{\sigma}=\fv{\sigma}{\Gamma_1}\cap\fv{\sigma}{\wit(\Gamma_2)}$ for every $\sigma\in \ssi$.

\end{theorem}

All three items of \Cref{yoniafterclark} include assumptions that guarantee that the two theories agree on cardinalities of shared sorts.
For example, in the first item, we first shrink the $S^{nsi}$-domains of the $T_2$-model using strong finite witnessability, and then expand them using smoothness. But then, to obtain infinite domains for the $S^{si}$ sorts, stable infiniteness is not enough, as we need to maintain the cardinalities of the $S^{nsi}$ domains while making the domains of the $S^{si}$ sorts infinite. For this, the stronger property of strong stable infiniteness is used.

The formal proof of this theorem is provided in \Cref{main-proof}, below.
\Cref{fig:proofs-plot} is a visualization of the claims in \Cref{yoniafterclark}.
The theorem considers two variants
of strong finite witnessability, two variants of smoothness,
and three variants of stable infiniteness.
For each of the three cases of \Cref{yoniafterclark}, \Cref{fig:proofs-plot} shows which variant of each property is assumed. 
The
height of each bar corresponds to the strength of the property.
In the first case, we use ordinary strong finite witnessability and smoothness, but the strongest
variant of stable infiniteness; 
in the second, we use ordinary strong finite witnessability with the new variants of stable infiniteness and smoothness;
and for the third, we use ordinary stable infiniteness and
the stronger variants of 
strong finite witnessability and smoothness.
The order of the bars corresponds to the order of their usage in the proof of each case.
%
%
The stage at which stable infiniteness is used
determines the required strength of the other properties: 
whatever is used before is taken in ordinary form, 
and whatever is used after requires
a stronger form.

\begin{figure}[t]
\begin{mdframed}
	\begin{center}
\begin{tikzpicture}
\begin{axis}[
xticklabels={,,Case 1,,Case 2,,,Case 3},
ytick={0,1,2,3},
ymin=0,
xmin=0,
yticklabels={,regular,medium,strong},
width=\textwidth,
height=.3\textwidth,
bar width=1,
legend style={at={(0.5,-0.3)},
	anchor=north,legend columns=-1}
]

\addplot [purple,fill,
ybar, area legend,
] coordinates {
(1,1) [s.f.w.]
};
\addplot [Green,fill,
ybar, area legend,
] coordinates {
(2,1) [sm]
};
\addplot [darkgray,fill,
ybar,area legend,
] coordinates {
(3,3) [s.i.]
};
\legend{strong finite witnessability, smoothness, stable infiniteness}

\addplot [purple,fill,
ybar, area legend
] coordinates {
(6,1) 
};
\addplot [darkgray,fill,
ybar,
] coordinates {
(7,2) 
};
\addplot [Green,fill,
ybar,
] coordinates {
(8,3) 
};

\addplot [darkgray,fill,
ybar,
] coordinates {
(11,1) 
};

\addplot [purple,fill,
ybar,
] coordinates {
(12,3) 
};

\addplot [Green,fill,
ybar,
] coordinates {
(13,3) 
};

\end{axis}

\end{tikzpicture}	
\caption{\Cref{yoniafterclark}.
The height of each bar corresponds to the strength of the property. The bars are ordered according to their usage in the proof.}
\label{fig:proofs-plot}
\end{center}
\end{mdframed}
\end{figure}

Going back to the standard definitions of stable infiniteness, smoothness, and strong finite witnessability, we get the following corollary by using case 1 of the theorem and noticing that smoothness w.r.t.~$S$ implies strong stable infiniteness w.r.t.~any partition of $S$.

\begin{corollary}
\label{item:fwpolite} 
Let $\ssi\subseteq S$ and $\snsi=S\setminus \ssi$.
Suppose $\Ta_1$ is stably infinite w.r.t.~$\ssi$ and
$\Ta_2$ is strongly finitely witnessable w.r.t.~$\snsi$ with witness $\wit$ and smooth w.r.t. 
$S$.
Then, the following are equivalent:
\begin{itemize*}
    \item[1.] $\Gamma_1\cup \Gamma_2$ is $(\T_1\oplus \T_2)$-satisfiable;
    \item[2.] there exists an arrangement $\delta_V$ over $V$ such that $\Gamma_1\cup \delta_V$ is $\Ta_1$-satisfiable and $\wit(\Gamma_2)\cup \delta_V$ is $\Ta_2$-satisfiable,
\end{itemize*}
where $V=\bigcup_{\sigma\in S}V_{\sigma}$, with
$V_{\sigma}=\fv{\sigma}{\wit(\Gamma_2)}$ for $\sigma\in \snsi$
and $V_{\sigma}=\fv{\sigma}{\Gamma_1}\cap\fv{\sigma}{\wit(\Gamma_2)}$ for $\sigma\in \ssi$.
\end{corollary}

Finally, the following result, which is closest
to \Cref{thm:original}, is
directly obtained
 from \Cref{item:fwpolite}, since
the strong politeness of $\Ta_{2}$ w.r.t.~$\ssi\cup \snsi$ implies that it is 
strongly finitely witnessable w.r.t.~$\snsi$ and smooth w.r.t.~$\ssi\cup \snsi$.

\begin{corollary}
\label{thm:main}
Let $\ssi\subseteq S$ and $\snsi=S\setminus \ssi$.
If $\Ta_1$ is stably infinite w.r.t.~$\ssi$ and
$\Ta_2$ is
strongly polite w.r.t.~$S$ with a witness
 $\wit$,
 then the following are equivalent:
\begin{itemize*}
    \item[1.] $\Gamma_1\cup \Gamma_2$ is $(\T_1\oplus \T_2)$-satisfiable;
    \item[2.] there exists an arrangement $\delta_V$ over $V$ such that $\Gamma_1\cup \delta_V$ is $\Ta_1$-satisfiable and $\wit(\Gamma_2)\cup \delta_V$ is $\Ta_2$-satisfiable,
\end{itemize*}
where $V=\bigcup_{\sigma\in S}V_{\sigma}$, with
$V_{\sigma}=\fv{\sigma}{\wit(\Gamma_2)}$ for each $\sigma\in \snsi$
and $V_{\sigma}=\fv{\sigma}{\Gamma_1}\cap\fv{\sigma}{\wit(\Gamma_2)}$ for each $\sigma\in \ssi$.

\end{corollary}


Compared to \Cref{thm:original}, \Cref{thm:main} partitions
$S$ into $\ssi$ and $\snsi$ and requires 
that $\Ta_1$ be stably infinite w.r.t.~$\ssi$.
The gain from this requirement is that the set $V_{\sigma}$ 
is potentially reduced for $\sigma\in \ssi$.
Note that unlike \Cref{yoniafterclark} and \Cref{item:fwpolite},
\Cref{thm:main} has the same assumptions regarding $\Ta_2$
as the original \Cref{thm:original} from \cite{JBLPAR}.
We show its potential impact in the next example.

\begin{example}
\label{example:arrbvint}
Consider the theory $\thlistintbvfour$ from \Cref{ex:thlistintbvfour}.
Let $\Gamma_{1}$ be $x=5\wedge v=0000\wedge w=w\ \&\ v$,
and let $\Gamma_2$ be $a_{0}=cons(x, v,a_1)\wedge \bigwedge_{i=1}^{n}a_{i}=cons(y_i,w,a_{i+1})$.
Using the witness function $\wit$ from~\cite{DBLP:conf/cade/0007ZRLFB20}, 
$\wit(\Gamma_2)=\Gamma_{2}$.
The polite combination approach reduces the ${\thlistintbvfour}$-satisfiability
of $\Gamma_{1}\wedge\Gamma_2$ to the existence of an 
arrangement $\delta$ over $\{x,v,w\}\cup\set{y_1\til y_n}$, such that
$\Gamma_{1}\wedge\delta$ is $\thintbvfour$-satisfiable
and $\wit(\Gamma_2)\wedge\delta$ is ${\thlist}$-satisfiable.
\Cref{thm:main} shows that we can do better.
Since ${\thintbvfour}$ is stably infinite w.r.t.~$\{\intsort\}$,
it is enough to check the existence of an arrangement 
over the variables of sort $\bvsort{4}$ that occur in $\wit(\Gamma_2)$, 
together with the variables of sort $\intsort$ that are shared between $\Gamma_{1}$ and $\Gamma_{2}$.
This means that arrangements over $\{x,v,w\}$ are considered, instead of over $\{x,v,w\}\cup\set{y_1\til y_n}$.
As $n$ becomes large, standard polite combination requires considering exponentially more arrangements, while the number of arrangements considered by our combination method remains the same.
\qed
\end{example}

\subsection{Proof of \Cref{yoniafterclark}}
\label{main-proof}
The left-to-right direction is straightforward, using the reducts
of the satisfying interpretation of $\Gamma_1\cup\Gamma_2$ to
$\Sigma_1$ and $\Sigma_2$.
We now focus on the right-to-left direction, and begin with the following lemma,
which strengthens \Cref{thm:ordered-sl}, obtaining a many-sorted L\"{o}wenheim-Skolem Theorem,
where the cardinality of the finite sorts remains the same.

\begin{lemma}
\label{lem:many-sorted-sl}
Let $\Sigma$ be	a signature, $\Ta$ a $\Sigma$-theory,
$\varphi$ a $\Sigma$-formula, and $\Aa$ a
$\Ta$-interpretation that satisfies $\phi$.
Let $\sorts{\Sigma}=S_{\Aa}^{\finite}\uplus S_{\Aa}^{\infinite}$,
where $\sigma^{\Aa}$ is finite for every
$\sigma\in S_{\Aa}^{\finite}$
and
$\sigma^{\Aa}$ is infinite for every
$\sigma\in S_{\Aa}^{\infinite}$.
Then there exists a $\Ta$-interpretation $\Ba$ that satisfies
$\varphi$ such that
$\card{\sigma^{\Ba}}=\card{\sigma^{\Aa}}$
for every $\sigma\in S_{\Aa}^{\finite}$
and
$\sigma^{\Ba}$ is countable for every $\sigma\in S_{\Aa}^{\infinite}$.
\end{lemma}

\noindent
The proof of \Cref{yoniafterclark} continues with the following main lemma.

\begin{lemma}[Main Lemma]
\label{lem:same-card}
Let $\ssi\subseteq S$ and $\snsi=S\setminus \ssi$,
Suppose $\Ta_1$ is stably infinite w.r.t.~$\ssi$ and
that one of the three cases of \Cref{yoniafterclark} holds.
Further, assume 
there exists an arrangement $\delta_V$ over $V$ such that $\Gamma_1\cup \delta_V$ is $\Ta_1$-satisfiable, and $\wit(\Gamma_2)\cup \delta_V$ is $\Ta_2$-satisfiable,
where $V=\bigcup_{\sigma\in S}V_{\sigma}$, with
$V_{\sigma}=\fv{\sigma}{\wit(\Gamma_2)}$ for each $\sigma\in \snsi$
and $V_{\sigma}=\fv{\sigma}{\Gamma_1}\cap\fv{\sigma}{\wit(\Gamma_2)}$ for each $\sigma\in \ssi$.
Then,
there is a $\Ta_1$-interpretation $\Aa$
that satisfies $\Gamma_1\cup \delta_{V}$ and a
$\Ta_2$-interpretation $\Ba$ that satisfies $\wit(\Gamma_2)\cup\delta_V$
such that
$\card{\sigma^{\Aa}}=\card{\sigma^{\Ba}}$ for all
$\sigma\in S$.
\end{lemma}

\begin{proof}
Let $\psi_2:=\wit(\Gamma_2)$.
Since $\Ta_1$ is stably infinite w.r.t.~$\ssi$, there is a $\Ta_1$-interpretation $\Aa$ 
satisfying $\Gamma_1\cup\delta_{V}$ in which
$\sigma^{\Aa}$ is infinite for each $\sigma\in \ssi$.
By \Cref{thm:ordered-sl}, we may assume that $\sigma^{\Aa}$ is countable
for each $\sigma\in \ssi$.
We consider the first case of \Cref{yoniafterclark}
(the others are omitted due to space constraints).
Suppose $\Ta_2$ is strongly stably infinite w.r.t.~$(\ssi,\snsi)$ and strongly polite w.r.t.~$\snsi$.
Since $\Ta_2$ is strongly finitely-witnessable w.r.t.~$\snsi$,
there exists a $\Ta_2$-interpretation $\Ba$
that satisfies $\psi_{2}\cup\delta_{V}$ such that
$\sigma^{\Ba}=V_{\sigma}^{\Ba}$ for each $\sigma\in \snsi$.
Since $\Aa$ and $\Ba$ satisfy $\delta_{V}$,
we have that for every $\sigma\in \snsi$,
$\card{\sigma^{\Ba}}=\card{V_{\sigma}^{\Ba}}=\card{V_{\sigma}^{\Aa}}\leq \card{\sigma^{\Aa}}$.
$\Ta_2$ is also smooth w.r.t.~$\snsi$, and so there exists
a $\Ta_2$-interpretation $\Ba'$ satisfying
$\psi_2\cup \delta_{V}$ such that
$\card{\sigma^{\Ba'}}=\card{\sigma^{\Aa}}$ for each $\sigma\in \snsi$.
Finally, $\Ta_2$ is strongly stably infinite w.r.t.~$(\ssi,\snsi)$,
so there is a $\Ta_2$-interpretation
$\Ba''$ that satisfies $\psi_{2}\cup\delta_{V}$ such that
$\sigma^{\Ba''}$ is infinite for each $\sigma\in \ssi$
and $\card{\sigma^{\Ba''}}=\card{\sigma^{\Ba'}}=\card{\sigma^{\Aa}}$
for each $\sigma\in \snsi$.
By \Cref{lem:many-sorted-sl}, we may assume that
$\sigma^{\Ba''}$ is countable for each $\sigma\in \ssi$.  Thus, $\card{\sigma^{\Ba''}}=\card{\sigma^{\Aa}}$ for each
$\sigma\in S$.
\end{proof}

%
%
%

We now conclude 
 \Cref{yoniafterclark}:
Let $\Ta:=\Ta_1\oplus\Ta_{2}$.
  \Cref{lem:same-card} 
  gives us 
  a $\Ta_1$ interpretation $\Aa$ with 
  $\Aa\models \Gamma_{1}\cup \delta_{V}$
  and a $\Ta_2$ interpretation $\Ba$ with
  $\Ba\models \psi_{2}\cup \delta_{V}$,
  and $\card{\sigma^{\Aa}}=\card{\sigma^{\Ba}}$
  for $\sigma\in S$.
  Set $\Gamma_{1}':=\Gamma_1\cup\delta_{V}$ and
  $\Gamma_{2}':=\psi_{2}\cup\delta_{V}$.
  Then, $V_{\sigma}=\fv{\sigma}{\Gamma_{1}'}\cap \fv{\sigma}{\Gamma_{2}'}$
  for $\sigma\in S$.
Now,
  $\Aa\models \Gamma_{1}'\cup \delta_{V}$ and
  $\Ba\models \Gamma_{2}'\cup \delta_{V}$.
  Also, $|\sigma^{\Aa}|=|\sigma^{\Ba}|$ for 
  $\sigma\in S$.
  By Theorem \ref{oldbutgood},
  $\Gamma_{1}'\cup\Gamma_{2}'$ is $\Ta$-satisfiable.
  In particular, $\Gamma_{1}\cup\{\psi_{2}\}$ is $\Ta$-satisfiable, and
  hence also $\Gamma_{1}\cup\{\exists \overline{w}.\psi_{2}\}$, with
  $\overline{w}=\fv{}{\wit(\Gamma_2)}\setminus\fv{}{\Gamma_2}$.
  Finally, $\exists\overline{w}.\wit(\Gamma_2)$ is $\Ta_2$-equivalent to $\Gamma_2$, hence
  $\Gamma_{1}\cup\Gamma_{2}$ is $\Ta$-satisfiable.
  \qed

\bigskip

\section{Preliminary Case Study}
\label{sec:casestudy}
The 
results
presented in \Cref{sec:polite_nelson_oppen}
was motivated by a set of smart contract verification benchmarks.
We obtained these benchmarks by applying the open-source Move Prover verifier~\cite{DBLP:conf/cav/ZhongCQGBPZBD20} to smart contracts found in the open-source Diem project~\cite{github-diem}.
The Move prover is a formal verifier for smart contracts
written in the Move language \cite{move_white} and was designed to target smart contracts used in the Diem blockchain \cite{libra_blockchain_white}.
It works via a translation to the Boogie 
verification framework \cite{thisisboogie2}, which in turn
produces \smtlib benchmarks that are dispatched to SMT solvers.
The benchmarks we obtained involve datatypes, integers, Booleans, and quantifiers.
Our case study began by running \cvcfour~\cite{CVC4} on the benchmarks.
For most of the benchmarks that were solved by \cvcfour,
theory combination took a small percentage of the overall runtime of the solver,
accounting for 10\% or less in all but 1 benchmark.
However, solving that benchmark
took 81 seconds, of which 20 seconds was dedicated to theory combination.

We implemented an optimization to the datatype solver of \cvcfour
based on \Cref{thm:main}.
With the original polite combination method, every term that
originates from the theory of datatypes with another sort
is shared with the other theories, triggering an analysis of the arrangements
of these terms.
In our optimization,
we limit the sharing of such terms
to those of Boolean sort.
In the language of \Cref{thm:main}, $\Ta_{1}$ is the combined theory
of Booleans, uninterpreted functions, and integers, which is stably infinite w.r.t.
the uninterpreted sorts and integer sorts. $\Ta_2$ is an instance of the theory of datatypes,
which is strongly polite w.r.t its element sorts, which in this case are 
the sorts of $\Ta_1$.

A comparison of an original and optimized run on the difficult benchmark is shown in \Cref{fig:challenge}.
As shown, the optimization reduces the total running time by 75\%, and
the time spent on theory combination in particular by 83\%. 
To further isolate the effectiveness of our optimization,
we report the number of terms that each theory solver considered.
In \cvcfour, 
constraints are not flattened, so shared \emph{terms} are processed instead of shared variables.
Each theory solver maintains its own data structure for tracking equality information. These data structures contain terms belonging to the theory that either come from the input assertions or are shared with another theory.
A data structure is also maintained that contains all shared terms belonging to any theory.
The last 4 columns of Figure 4 count the number of times (in thousands) a term was added to the equality data structure for the theory of datatypes (DT), integers (INT), and uninterpreted functions and Booleans (UFB), as well as to the the shared term data structure (shared).
With the optimization, the datatype solver 
keeps more inferred assertions internally, which leads to an increase in the number of additions of terms to
its data structure.
However, sharing fewer terms, reduces the number of terms in the data structures for the other theories.
Moreover, while the total number of terms considered remains roughly the same, the number of shared terms decreases by 24\%.
This suggests that although the workload on the individual theory solvers is roughly similar,
a decrease in the number of shared terms in the optimized run results in a significant improvement in the overall runtime.
Although our evidence is only anecdotal at the moment, we believe this benchmark is highly representative of the potential benefits of our optimization.

\begin{figure}[t]
\begin{mdframed}
\begin{center}
\begin{tabular}{|c|c|c||c|c|c|c|c|c|}\hline
& total (s) & comb (s) & DT & INT  & UFB & shared \\\hline
optimized &     34.9 & 3.4  & 236.1 & 212.1 & 78.4 & 125.8\\\hline
original & 81.5 & 20.3 & 116.0 & 281.0 & 123.9 & 163.5\\\hline
\end{tabular}
	
\end{center}
\caption{Runtimes (in seconds) and number of terms (in thousands) added to the data structures of DT, INT, UFB, and the number of shared terms (shared).}
\label{fig:challenge}	
\end{mdframed}
	
\end{figure}

\section{Conclusion}
\label{sec:conc}
This paper makes two contributions:
First, we separated politeness and strong politeness,
which shows that sometimes, the (typically harder) task 
of finding a strong witness is not a waste of efforts.
Then, we provided an optimization to the polite combination method,
which is applies when one of the theories in the combination is stably infinite w.r.t a subset of the sorts.

We envision several directions for future work. First, 
the sepration of politeness from strong politeness
demonstrates a need to identify sufficient criteria for the equivalence
of these notions --- such as, for instance, the {\em additivity} criterion introduced 
by Sheng et al.~\cite{DBLP:conf/cade/0007ZRLFB20}.
%
Second, polite combination might be optimized
by applying the witness function only to part of the purified input formula.
Finally, we plan to extend the initial
implementation of this approach in \cvcfour
and evaluate its impact based on more benchmarks. 

\newpage

\bibliographystyle{splncs04}
\bibliography{polite}

\begin{report}
\newpage
\appendix
\section{Appendix}

\subsection{Theories vs. Classes of Structures}
\label{ap:sticky}
In papers about theory combination, theories are often defined
in terms of some set $Ax$ of sentences (\emph{axioms}) 
(see, e.g., \cite{Casal2018,TinZar-JELIA-04,JBLPAR}). 
Specifically, a theory is defined as the set of all sentences entailed
by $Ax$ or, interchangeably, as the class of 
all structures that satisfy $Ax$.
This is the approach we take in this paper.
The main reason for this is that the combination theorems we prove
and cite here rely on some forms of the
L\"{o}wenheim-Skolem theorem, which do not hold for arbitrary
classes of structures, but do hold when defining theories this way.
On the other hand, theories in 
the \smtlib standard, as well as in many SMT papers about
individual theories,
are defined more generally as classes of structures without 
reference to a set of axioms.

However, this discrepancy is not substantial since 
the two notions of a theory as a class of structures are easily
interreducible;
as mentioned in the introduction,
every theory $T$ in the second, more general sense induces 
a theory in the first sense
that is equivalent to $T$ for all of our intents and purposes 
since it entails exactly the same sentences as $T$.
To be more precise, the combination theorems that we prove and cite
only regard satisfiability of formulas in a theory 
(though their proofs
may analyze the structures of a theory).
The important thing is that the transformation between the two notions
preserves satisfiability, and therefore interchanging these notions
can be done without loss of generality.
For completeness, we prove this fact below:

\begin{lemma}
Let $\Sigma$ be a signature, $\Ca$ a class of $\Sigma$-structures,
$Ax$ the set of $\Sigma$-sentences satisfied by all structures of $\Ca$,
and $\Ta_{\Ca}$ the class of all $\Sigma$-structures 
that satisfy all sentences of $Ax$.
Then, for every
$\Sigma$-formula $\varphi$, 
$\varphi$ is $\Ta_{\Ca}$-satisfiable iff $\varphi$ is satisfied
by some $\Sigma$-interpretation whose variable-free part is
in $\Ca$.
\end{lemma}

\begin{proof}
Every interpretation whose variable-free part is in $\Ca$ is
a $\Ta_{\Ca}$-interpretation,
and so the right-to-left direction
trivially holds.
Now, suppose $\varphi$ is not satisfied by any $\Sigma$-interpretation
whose variable-free part is in $\Ca$.
Then its existential closure $\exists\overline{x}.\varphi$ is
not satisfied by any structure of $\Ca$, and hence
$\neg\exists\overline{x}.\varphi\in Ax$.
\emph{Ad absurdum}, suppose that $\varphi$ is $\Ta_{\Ca}$-satisfiable.
Then there is a $\Ta_{\Ca}$-interpretation $\Aa$
such that $\Aa\models\varphi$.
In particular, $\Aa\models \exists \overline{x}.\varphi$.
But since $\Aa$ is a $\Ta_{\Ca}$-interpretation, we must also have
$\Aa\models\neg\exists\overline{x}.\varphi$, which is a contradiction.
\end{proof}

\subsection{Proof of \Cref{lem:ta23smooth}}

Let $\phi$ be a quantifier-free $\Sigma_{2}$-formula,
$\Aa$ a $\Ta_{2,3}$-interpretation that satisfies $\phi$ and
$\kappa$ a function from $\set{\sigma_{1},\sigma_{2}}$ to the class
of cardinals such that
$\kappa(\sigma_{1})\geq\card{\sigma_{1}^{\Aa}}$ and
$\kappa(\sigma_{2})\geq\card{\sigma_{2}^{\Aa}}$.
We construct a $\Sigma_{2}$-interpretation $\Aa'$ as follows.
For $i\in\set{1,2}$, we let
$\sigma_{i}^{\Aa'}:=\sigma_{i}^{\Aa}\uplus B$ for some
set $\Ba$ of countable cardinality if
$\kappa(\sigma_{i})$ is infinite
or of cardinality $\kappa(\sigma_{i})-\card{\sigma_{i}}^{\Aa}$ otherwise.
Notice that this is well defined because
$\kappa(\sigma_{i})\geq\card{\sigma_{i}^{\Aa}}$.
As for variables, $x^{\A'}:=x^{\Aa}$ for each variable in
$\fv{}{\phi}$. 
This is well defined because the domains of $\sigma_{1}$ and $\sigma_{2}$ 
were only possibly extended, not reduced.
First, we prove that $\Aa'$ is a $\Ta_{2,3}$-interpretation.
If $\kappa(\sigma_{1})=2$, then since 
$\kappa(\sigma_{1})\geq\card{\sigma_{1}^{\Aa}}$, we must have
that $\card{\sigma_{1}^{\Aa}}=2$, which means that
$\card{\sigma_{2}}^{\Aa}$ is infinite, which in turn means
that $\kappa(\sigma_{2})$ is infinite as well.
Hence in this case we have
$\card{\sigma_{1}^{\Aa'}}=\kappa(\sigma_{1})=2$ 
and
$\card{\sigma_{2}^{\Aa'}}=\kappa(\sigma_{2})=\infty$.
Otherwise, $\kappa(\sigma_{1})\geq 3$, and hence 
$\card{\sigma_{1}^{\Aa'}}=\kappa(\sigma_{1})\geq 3$ and also
$\card{\sigma_{2}^{\Aa'}}=\kappa(\sigma_{2})\geq\card{\sigma_{2}^{\Aa}}\geq 3$.
Clearly, $\Aa'$ satisfies $\phi$ as the interpretations of variables did not change.
Finally,
$\card{\sigma_{1}^{\Aa'}}=\kappa(\sigma_{1})$ 
and
$\card{\sigma_{2}^{\Aa'}}=\kappa(\sigma_{2})$ by construction.

\qed

\subsection{Proof of \Cref{lem:ta23fw}}

Define a function $\wit$ by 
$\wit(phi):=phi\wedge x_1=x_1\w x_2=x_2\w x_3=x_3\w  y_1=y_1\w y_2=y_2\w y_3=y_3$ for
fresh variables $x_1$, $x_2$ and $x_3$ of sort $\sigma_{1}$ and
$y_1$, $y_2$ and $y_3$ of sort $\sigma_{2}$.
We prove that $\wit$ is a witness for $\Ta_{2,3}$ w.r.t.~$\set{\sigma_1,\sigma_2}$.
 $\phi$ and $\exists x_1,x_2,x_3,y_1,y_2,y_3. \wit(\phi)$ are trivially 
logically equivalent and in particular $\Ta_{2,3}$-equivalent.
 We prove that $\wit(\phi)$ is finitely witnessed for $\Ta_{2,3}$ w.r.t.
$\set{\sigma_{1},\sigma_{2}}$.
Suppose that $\wit(\phi)$ is $\Ta_{2,3}$-satisfiable and let
$\Aa$ be a satisfying $\Ta_{2,3}$-interpretation.
Define a $\Sigma_{2}$-interpretation $\Ba$ simply by
$\sigma_{1}^{\Ba}=\fv{\sigma_{1}}{\phi}^{\Aa}\uplus\set{a_1,a_2,a_3}$ and
$\sigma_{2}^{\Ba}=\fv{\sigma_{2}}{\phi}^{\Aa}\uplus\set{b_1,b_2,b_3}$ 
for $a_1,a_2,a_3\notin\sigma_{1}^{\Aa}$
and
$b_1,b_2,b_3\notin\sigma_{2}^{\Aa}$.
The interpretations of variables from $\phi$ are the same as in $\Aa$.
As for the fresh variables $x_i^{\Ba}:=a_i$ and $y_i^{\Ba}:=b_i$
for $i\in\set{1,2,3}$.
We prove that $\Ba$ finitely witnesses $\wit(\phi)$ for $\Ta_{2,3}$ w.r.t.
$\set{\sigma_1,\sigma_2}$.
First, $\Ba$ is a $\Ta_{2,3}$-interpretation, as by construction
$\card{\sigma_{1}^{\Ba}},\card{\sigma_{2}^{\Ba}}\geq 3$.
Second, $\Ba\models\phi$ as the interpretations of variables from $\phi$
did not change, and trivially satisfies the new identities, and so
$\Ba\models \wit(\phi)$.
Third, by construction
$\sigma_{1}^{\Ba}=\fv{\sigma_{1}}{\phi}^{\Aa}\uplus\set{a_1,a_2,a_3}=
\fv{\sigma_{1}}{\phi}^{\Ba}\uplus\set{x_1^{\Ba},x_2^{\Ba},x_3^{\Ba}}=
\fv{\sigma_{1}}{\wit(\phi)}^{\Ba}
$, and similarly for $\sigma_{2}$.

\qed

\subsection{Proof of \Cref{lem:23nonfw}}
Let $\wit$ be a witness for $\Ta_{2,3}$ w.r.t.~$\set{\sigma_1,\sigma_2}$.
We show that it is not strong.
In particular, we show that $\wit(v=v)$ is not strongly finitely witnessed
for $\Ta_{2,3}$ w.r.t.~$\set{\sigma_1,\sigma_2}$.
Consider a $\Ta_{2,3}$-interpretation $\Aa$
with $\card{\sigma_{1}^{\Aa}}=2$ and $\card{\sigma_{2}^{\Aa}}=\aleph_0$.
Clearly, $\Aa\models v=v$, and so
$\Aa\models \exists\,\overline{w}.\: \wit(v=v)$, with $\overline{w}$ being
the variables in $\wit(v=v)$ other than $v$.
This in turn means that there is a $\Ta_{2,3}$-interpretation 
$\Aa'$ that satisfies $\wit(v=v)$, different from $\Aa$ only
in the interpretations of $\overline{w}$, if anywhere.
Let $\delta$ be the arrangement over $\fv{}{\wit(v=v)}$ induced by
$\Aa'$.
Then, $\delta$ either asserts that all variables in  $\fv{\sigma_{1}}{\wit(v=v)}$ 
are identical, or it partitions them into two equivalence classes.
$\Aa'\models \wit(v=v)\w \delta$, and so
$\wit(v=v)\w \delta$ is $\Ta_{2,3}$-satisfiable.
We show that it  does not have a finite witness for $\Ta_{2,3}$ w.r.t.~$S$.
Suppose for contradiction that $\Ba$ is a finite 
witness of $\wit(v=v)\w\delta$ for $\Ta_{2,3}$ w.r.t.~$S$.
Then $\card{\sigma_{1}^{\Ba}}=\card{\fv{\sigma_{1}}{\wit(v=v)\w\delta}^{\Ba}}$.
Now, $\Ba\models\delta$ and $\Ba$ is a $\Ta_{2,3}$-interpretation, meaning $\card{\sigma_1^{\Ba}}\ge 2$, so if $\delta$ requires all variables of sort $\sigma_1$ to be equal, we already have a contradiction.  On the other hand, if $\delta$ partitions the variables into two equivalence classes,
we get that $\card{\sigma_{1}^{\Ba}}=2$.
But since $\Ba$ finitely witnesses $\wit(v=v)\w\delta$ for $\Ta_{2,3}$ w.r.t.
$\set{\sigma_1,\sigma_2}$,
we also get that $\sigma_2^{\Ba}$ is finite, meaning
$\Ba$ is not a $\Ta_{2,3}$-interpretation.
\qed

\subsection{Proof of \Cref{lem:smoothatleast}}
Let $\Aa$ be the $\Ta$-structure with a minimal number of elements,
and let $n=\card{\sigma^{\Aa}}$.
To show that every $\Sigma_0$-structure that satisfies $\psi_{\geq n}^{\sigma}$ belongs to $\Ta$,
let $\Ba$ be a $\Sigma_0$-structure that
satisfies $\psi_{\geq n}^{\sigma}$ and let
$m$ be the cardinality of $\sigma^{\Ba}$.
Then $m\geq n$. 
Clearly, $\Aa\models x=x$ and has $n$ elements.
Since $\Ta$ is smooth w.r.t.~$\sigma$,
there exists a $\Ta$-interpretation (that satisfies $x=x$)
with cardinality $m$. This interpretation must be $\Ba$,
as the lack of any symbols means that the only thing that
distinguishes between $\Sigma_0$-structures is their cardinality (modulo isomorphism).
For the converse, note that by the choice of $n$ as minimal,
every $\Ta$-structure satisfies $\psi_{\geq n}^{\sigma}$.
\qed

\subsection{Proof of \Cref{prop:singlesort}}
$x=x$ is clearly $\Ta$-satisfiable.
Since $\Ta$ is finitely witnessable (say with witness $\wit$),
there is a $\Ta$-interpretation $\Aa$ that satisfies
$\wit(x=x)$ such that $\sigma^{\Aa}$ is finite.
$\Ta$ is smooth, and hence, by \Cref{lem:smoothatleast},
it is axiomatized by $\psi_{\geq n}^{\sigma}$ for some $n$.
Define $\wit'(\phi):=\phi \wedge \distinct(x_1\til x_n)$ for fresh
$x_1\til x_n$.
Since $\Ta$ is axiomatized by $\psi_{\geq n}^{\sigma}$,
$\phi$ is $\Ta$-equivalent to $\exists \overline{x}.\wit'(\phi)$.
Further, for any arrangement $\delta$ over some set of variables,
and any $\Ta$-interpretation $\Aa'$ that satisfies $\wit'(\phi)\wedge\delta$, if the domain of $\Aa'$ is reduced to contain only the elements in $\fv{}{\wit'(\phi)\wedge\delta}^{\Aa'}$, the result is still a $\Ta$-interpretation since $\wit'(\phi)$ contains
$\distinct(x_1\til x_n)$.
We therefore get that $\wit'$ is a strong witness for $\Ta$ w.r.t.~$\sigma$.
\qed

\subsection{Proof of \Cref{lem:notsmooth}}
Let $\phi$ be $x=x$ and
$\Aa$ be a $\Sigma$-interpretation with $\sigma^{\Aa}=\set{1,2}$
and $x^{\Aa}=1$.
Then $\Aa$ is a $\TaEven$-interpretation that satisfies $\phi$.
Let $\kappa$ defined by $\kappa(s)=3$.
Then $3=\kappa(s)\geq\card{\sigma^{\Aa}}=2$.
However, there is no $\Sigma$-interpretation $\Aa'$
with $\card{\sigma^{\Aa'}}=3$.

\qed

\subsection{Proof of \Cref{thm:weakeven}}
Define $\wit(\phi)$ as follows.
Let
$E$ be the set of all equivalence relations over 
$\fv{}{\phi}\cup\{w\}$ for some fresh variable $w$.  
Let $\even(E)$ be the set of all equivalence relations in $E$ for which the number of equivalence classes is even.
Then, $\wit(\phi)$ is
$\phi\wedge\bigvee_{e\in \even(E)}\delta_e$, where for an equivalence relation
$e\in \even(E)$, $\delta_e$ is the arrangement induced by $e$:
\begin{equation*}
\bigwedge_{(x,y)\in e}x=y\wedge\bigwedge_{x,y\in
\fv{}{\phi}\cup\set{w}\wedge(x,y)\not\in e}x\neq y
\end{equation*}
We prove that $\wit$ is a witness.
Let $\phi$ be a $\Sigma$-formula.
We first prove that it is $\TaEven$-equivalent to 
$\exists\,w.\:\wit(\phi)$.
Since $\phi$ is a conjunct of $\wit(\phi)$ that does
not include $w$,
every $\Aa$-interpretation that satisfies $\wit(\phi)$
also satisfies $\phi$. For the other direction, 
let $\Aa$ be a $\TaEven$-interpretation satisfying $\phi$.
Even though $\Aa$ may have infinitely many elements, the number of
elements in 
$\fv{}{\phi}^{\Aa}$ must be finite.
If the number of elements in $\fv{}{\phi}^{\Aa}$ is even,
then let $a$ be some arbitrary element of $\fv{}{\phi}^{\Aa}$.
Otherwise, let $a$ be an element in $\Aa$ different from all the elements in $\fv{}{\phi}^{\Aa}$ (there must be such an element since $\Aa$ has an even or infinite number of elements).  In either case, the number of elements in $(\fv{}{\phi}\cup\{w\})^{\Aa}$ is even.  Thus, if we modify $\Aa$ to map $w$ to $a$, then it must satisfy one of the disjuncts in $\wit(\phi)$.  Hence, $\Aa$ satisfies $\exists\,w.\:\wit(\phi)$.

Next, if $\wit(\phi)$ is $\TaEven$-satisfiable,
then there is a satisfying $\TaEven$-\hspace{0pt}interpretation $\Aa$ satisfying it.
$\Aa$ must satisfy one of the disjuncts in $\wit(\phi)$, which means
$\card{\fv{}{\wit(\phi)}^{\Aa}}$ is even. 
The restriction of $\Aa$ to $\fv{}{\wit(\phi)}^{\Aa}$ is a 
$\TaEven$-\hspace{0pt}interpretation that finitely witnesses $\wit(\phi)$.
\qed

\subsection{Proof of \Cref{thm:weaknotstrong}}
Let $\wit: \qf(\Sigma_0)\ra\qf(\Sigma_0)$.
We prove that $\wit$ is not a strong witness for $\TaEven$
w.r.t.~$\sigma$,
by showing that
$\wit(x=x)$ is not strongly finitely witnessed 
for $\TaEven$ w.r.t.~$\sigma$.
Consider a $\TaEven$-interpretation $\Aa$ with 2 elements,
which interprets all the variables in
$\fv{}{\wit(x=x)}$.
Clearly, $\Aa\models x=x$, and therefore, 
$\Aa\models \exists\,\overline{w}.\:\wit(x=x)$, where $\overline{w}$
is
$\fv{}{\wit(x=x)}\setminus{\set{x}}$.
Hence, there exists a $\TaEven$-interpretation $\Aa'$,
identical to $\Aa$, except possibly in its interpretation
of variables in $\fv{}{\wit(x=x)}\setminus{\set{x}}$, that satisfies
$\wit(x=x)$.
In particular, $\Aa'$ has two elements.
Let $\delta_{\Aa'}$ be the arrangement over
$\fv{}{\wit(x=x)}$ satisfied by $\Aa'$.
Then $\delta_{\Aa'}$ induces an equivalence relation with either 1 or 2 equivalence classes.
Let $v$ be a variable not in $\fv{}{\wit(x=x)}$.
Define an arrangement $\delta$ over $\fv{}{\wit(x=x)}\cup \set{v}$ as follows:
If $\delta_{\Aa'}$ induces one equivalence class, 
$\delta:=\delta_{\Aa'}\wedge \bigwedge_{u\in\fv{}{\wit(x=x)}}v=u$.
Otherwise,
$\delta:=\delta_{\Aa'}\wedge \bigwedge_{u\in\fv{}{\wit(x=x)}}v\neq u$.
In the first case, $\delta$ induces one equivalence class, and in the second, 
three.
$\wit(x=x)\wedge\delta$ does not have a finite witness for $\TaEven$
w.r.t.~$\sigma$, as any interpretation $\Ba$
that finitely witnesses it has either 1 or 3 elements, and hence it
is not in $\TaEven$.

\subsection{Proof of \Cref{item:fwpolite}}

$\Ta_2$ is smooth w.r.t.~$\ssi\cup \snsi$.
In particular, it is smooth w.r.t.~$\snsi$.
 We show that 
it is also strongly stably infinite w.r.t.
$(\ssi,\snsi)$, and then the result follows 
from case 1 of \Cref{yoniafterclark}.
Let $\phi$ be a $\Sigma$-formula,
$\Aa$ a
$\Ta$-interpretation that satisfies $\phi$.
Define $\kappa(\sigma)$ to be $\aleph_0$ for every
$\sigma\in \ssi$ such that $\sigma^{\Aa}$ is finite,
$\kappa(\sigma)=\card{\sigma^{\Aa}}$ for every
$\sigma\in \ssi$ such that $\sigma^{\Aa}$ is infinite,
and $\kappa(\sigma)=\card{\sigma^{\Aa}}$ for
every $\sigma\in \snsi$.
Since $\Ta$ is smooth w.r.t.~$\ssi\cup \snsi$,
there exists a $\Ta$-interpretation $\Ba$ that satisfies
$\phi$ with $\card{\sigma^{\Ba}}=\kappa(\sigma)$
for every $\sigma\in \ssi$ 
and $\card{\sigma^{\Ba}}=\kappa(\sigma)=\card{\sigma^{\Aa}}$
for every $\sigma\in \snsi$.

\qed

%

\subsection{Proof of \Cref{lem:many-sorted-sl}}

Let $Ax$ be the set of sentences that are satisfied by every $\Ta$-structure.
Define the following sets, based on formulas that are defined in \Cref{fig:cardfor}:
\begin{gather*}
\finite_{\Aa}:=\set{\psi_{=\card{\sigma^{\Aa}}}^{\sigma}\mid \sigma\in S_{\Aa}^{\finite}} \\
\infinite_{\Aa}:=\set{\neg\psi_{=n}^{\sigma}\mid \sigma\in S_{\Aa}^{\infinite},n\in\mathbb{N}} \\
A:=Ax\cup \finite_{\Aa}\cup \infinite_{\Aa}\cup\set{\phi} \\
\end{gather*}
Clearly, $\Aa\models A$.
By \Cref{thm:ordered-sl},
there exists a $\Sigma$-interpretation $\Ba$
that satisfies $A$ in which $\sigma^{\Ba}$ is countable
whenever it is infinite, for every $\sigma\in\sorts{\Sigma}$.
This in particular holds for every $\sigma\in S_{\Aa}^{\infinite}$.
Now let $\sigma\in S_{\Aa}^{\finite}$, then since $\Ba\models \finite_{\Aa}$,
$\card{\sigma^{\Ba}}=\card{\sigma^{\Aa}}$.
Finally, $\Ba\models\phi$ and it is a $\Ta$-interpretation.

\qed

\subsection{Remaining Cases in The Proof of \Cref{lem:same-card}}
Let $\psi_2:=\wit(\Gamma_2)$.
Since $\Ta_1$ is stably infinite w.r.t.~$\ssi$, there is a $\Ta_1$-interpretation $\Aa$ 
satisfying $\Gamma_1\cup\delta_{V}$ in which
$\sigma^{\Aa}$ is infinite for each $\sigma\in \ssi$.
By \Cref{thm:ordered-sl}, we may assume that $\sigma^{\Aa}$ is countable
for each $\sigma\in \ssi$.

\begin{itemize}
\item[Case 2]: Suppose $\Ta_2$ is stably infinite w.r.t
$(\ssi, \snsi)$,
smooth w.r.t.~$(\snsi,\ssi)$,
and strongly finitely witnessable w.r.t.~$\snsi$.
Then, there exists a $\Ta_2$-\hspace{0pt}interpretation $\Ba$
that satisfies $\psi_{2}\cup\delta_{V}$ such that
$\sigma^{\Ba}=V_{\sigma}^{\Ba}$ for every $\sigma\in \snsi$.
Since $\Aa$ and $\Ba$ satisfy $\delta_{V}$,
we have that for every $\sigma\in \snsi$,
$\card{\sigma^{\Ba}}=\card{V_{\sigma}^{\Ba}}=\card{V_{\sigma}^{\Aa}}\leq \card{\sigma^{\Aa}}$.
$\Ta_2$ is stably infinite w.r.t.~$(\ssi,\snsi)$,
and so there exists
a $\Ta_2$-interpretation $\Ba'$
that satisfies $\psi_2\cup \delta_{V}$
such that $\sigma^{\Ba'}$ is infinite for every
$\sigma\in \ssi$
and $\card{\sigma^{\Ba'}}\leq\card{\sigma^{\Ba}}\leq \card{\sigma^{\Aa}}$ for every $\sigma\in \snsi$.
$\Ta_2$ is smooth w.r.t.~$(\snsi,\ssi)$
and so there is a $\Ta_2$-interpretation $\Ba''$
satisfying $\psi_2\cup \delta_{V}$
such that
$\card{\sigma^{\Ba''}}=\card{\sigma^{\Aa}}$ for 
every $\sigma\in \snsi$
and $\card{\sigma^{\Ba''}}$ is infinite for every
$\sigma\in \ssi$.
Using \cref{lem:many-sorted-sl},
we may assume $\sigma^{\Ba''}$ is countable
and hence $\card{\sigma^{\Ba''}}=\card{\sigma^{\Aa}}$ for every
$\sigma\in S$.

\item[Case 3]: Suppose $\Ta_2$ is
stably infinite w.r.t.~$\ssi$,
smooth w.r.t.~$(\snsi, \ssi)$, and
strongly finitely witnessable w.r.t.~$(\snsi, \ssi)$.
Since it is stably infinite w.r.t.~$\ssi$,
there exists a $\Ta_2$-interpretation $\Ba$
that satisfies $\psi_{2}\cup \delta_{V}$ such that
$\sigma^{\Ba}$ is infinite for every $\sigma\in \ssi$.
$\Ta_2$ is strongly finitely-witnessable w.r.t.~$(\snsi,\ssi)$,
and hence there exists a $\Ta_2$-interpretation $\Ba'$
that satisfies $\psi_{2}\cup\delta_{V}$ such that
$\sigma^{\Ba'}=V_{\sigma}^{\Ba'}$ for every $\sigma\in \snsi$
and $\card{\sigma^{\Ba'}}$ is infinite for every
$\sigma\in \ssi$.
Since $\Aa$ and $\Ba'$ satisfy $\delta_{V}$,
we have that for every $\sigma\in \snsi$,
$\card{\sigma^{\Ba'}}=\card{V_{\sigma}^{\Ba'}}=\card{V_{\sigma}^{\Aa}}\leq \card{\sigma^{\Aa}}$.
$\Ta_2$ is smooth w.r.t.~$(\snsi,\ssi)$,
and so there exists
a $\Ta_2$-interpretation
$\Ba''$ that satisfies $\psi_2 \cup \delta_{V}$
such that $\card{\sigma^{\Ba''}}=\card{\sigma^{\Aa}}$ for every
$\sigma\in \snsi$ and
$\card{\sigma^{\Ba''}}$ is infinite for every
$\sigma\in \ssi$.
By \Cref{lem:many-sorted-sl},
we may assume that $\sigma^{\Ba''}$ is countable
for every $\sigma\in \ssi$, with the same cardinalities for sorts of $\snsi$,
and so we have $\card{\sigma^{\Ba''}}=\card{\sigma^{\Aa}}$
also for every $\sigma\in S$.

\end{itemize}	
\end{report}
\end{document}